\documentclass[conference,compsoc]{arxiv/report/IEEEtran}

\usepackage{xifthen}

\makeatletter
\newcommand\ifBeamerThenElse[2]{\@ifclassloaded{beamer}{#1}{#2}}
\makeatother

\usepackage[T1]{fontenc}

\usepackage{pifont}

\usepackage{import}

\usepackage{pdfpages}

\usepackage{tocloft}
\usepackage{enumitem}

\usepackage{amsthm}
\usepackage{amsfonts}
\usepackage{balance}

\usepackage{bm}

\usepackage{siunitx}

\usepackage[
    lambda,
    operators,
    sets,
]{cryptocode}

\usepackage{graphicx}

\usepackage{tikz}

\usepackage{tcolorbox}

\usepackage{transparent}

\usepackage{bclogo}

\usepackage{circledsteps}

\usepackage{subcaption}
\usepackage{tabularx}
\usepackage{booktabs}
\usepackage{threeparttable}
\usepackage{multirow}
\usepackage{listings}
\usepackage{float}

\usepackage{xurl}
\usepackage{hyperref}
\usepackage[htt]{hyphenat}
\usepackage{csquotes}
\usepackage[normalem]{ulem}
\usepackage{todonotes}
\usepackage{framed}

\usepackage[style=ieee]{biblatex}

\usepackage[capitalise, noabbrev]{cleveref}

\ifBeamerThenElse{}{\theoremstyle{definition}
\newtheorem{definition}{Definition}
\newtheorem{proposition}{Proposition}
\newtheorem{observation}{Observation}

}

\definecolor{NiceRed}{HTML}{B00008}
\definecolor{NiceGreen}{HTML}{2E7E2A}
\definecolor{NiceBlue}{HTML}{131877}
\definecolor{NiceMagenta}{HTML}{8A0087}
\definecolor{NiceOlive}{HTML}{727500}
\definecolor{NiceCyan}{HTML}{137776}

\definecolor{Starlight}{HTML}{a0a078}
\definecolor{PalePink}{HTML}{796071}
\definecolor{RoseGold}{HTML}{C29089}
\definecolor{PaleBlue}{HTML}{7D98B4}
\definecolor{PaleGreen}{HTML}{6DA478}

\hypersetup{
    linktoc=all,
    breaklinks=true,
    colorlinks=true,
    allcolors=black
}

\ifBeamerThenElse{}{}

\renewbibmacro*{doi+eprint+url}{\iftoggle{bbx:url}{\iffieldundef{doi}{\usebibmacro{url+urldate}}{}}{}\newunit\newblock \iftoggle{bbx:eprint}{\usebibmacro{eprint}}{}\newunit\newblock \iftoggle{bbx:doi}{\printfield{doi}}{}}

\DefineBibliographyStrings{english}{backrefpage = {cited on page},
  backrefpages = {cited on pages},
}

\renewcommand\ZZ{\mathbb{Z}}

\renewcommand\GG{\mathbb{G}}
\renewcommand\Pr{\mathrm{Pr}}

\newcommand*{\emptyCircle}{\begin{tikzpicture}[scale=0.11]\draw (0,0) circle (1);
    \end{tikzpicture}}
\newcommand*{\halfCircle}{\begin{tikzpicture}[scale=0.11]\draw (0,0) circle (1);
    \fill (0,0) -- (90:1) arc (90:90-50*3.6:1) -- cycle;
    \end{tikzpicture}}
\newcommand*{\fullCircle}{\begin{tikzpicture}[scale=0.11]\draw (0,0) circle (1);
    \fill (0,0) -- (90:1) arc (90:90-100*3.6:1) -- cycle;
    \end{tikzpicture}}

\newcommand{\bigPr}[2]{\Pr\left[\begin{array}{l}#1\end{array}
    \middle\vert\ 
    \begin{array}{l}#2\end{array}\right]
}

\newcommand{\up}[1]{^{(#1)}}
\newcommand{\down}[1]{_{#1}}

\newcommand\specialfont[1]{\textsf{#1}}

\newcommand{\subheading}[1]{\vspace{0.5 em}\noindent \textbf{#1}}

\newcommand{\missingBlocks}{m}
\newcommand{\m}{\missingBlocks}
\newcommand{\committeeSize}{n}
\newcommand{\startHeight}{\specialfont{h1}}
\newcommand{\heightOld}{\startHeight}
\newcommand\startEpoch{\specialfont{e1}}
\newcommand{\currentEpoch}{\specialfont{e2}}
\newcommand{\heightNew}{\specialfont{h2}}
\newcommand{\currentHeight}{\heightNew}
\newcommand{\quorumSize}{t}
\newcommand{\quorum}{\mathcal{Q}}

\newcommand{\groupHash}{H}
\newcommand{\step}[1]{\textcircled{#1}}
\newcommand{\apk}{\specialfont{apk}}
\newcommand{\ssk}{\specialfont{ssk}}
\newcommand{\sumKey}{\ssk}

\def\NextKeys/{\specialfont{NextK}}
\newcommand{\Nextkeys}{\specialfont{NextK}}
\newcommand{\Hdr}{\specialfont{Hdr}}
\newcommand{\newestHeader}{\Hdr}
\newcommand{\bp}{j}

\newcommand{\KeyGen}{\specialfont{KeyGen}}
\newcommand{\Sign}{\specialfont{Sign}}
\newcommand{\Verify}{\specialfont{Verify}}
\newcommand{\Aggregate}{\specialfont{Aggregate}}
\newcommand{\sk}{\specialfont{sk}}
\newcommand{\pk}{\specialfont{pk}}
\newcommand{\PK}{\specialfont{PK}}
\newcommand{\SK}{\specialfont{SK}}
\newcommand{\sig}{s}
\newcommand{\Sig}{S}

\def\VerifyAggregate/{\specialfont{Verify}}
\def\Agg/{\Aggregate/}
\def\VerifyAgg/{\VerifyAggregate/}
\def\PoP/{\specialfont{PoP}}
\def\VerifyPoP/{\specialfont{VerifyPoP}}
\def\Eval/{\Aggregate/}

\newcommand{\Update}{\specialfont{LC.Update}}
\newcommand{\VerifyState}{\specialfont{LC.Verify}}

\newcommand{\corresponds}{\specialfont{LCstate}}

\newcommand{\lcState}{d}
\newcommand{\lcStateOld}{\lcState\up{\heightOld}}

\newcommand{\sOld}{\lcStateOld}
\newcommand{\sNew}{\lcState\up{\heightNew}}
\newcommand{\validatorsAdv}{\validators_{\mathcal{A}}}

\newcommand{\validatorsThreshold}{f}
\newcommand{\FNClient}{\mathcal{F}}
\newcommand{\FullNode}{\ensuremath{F}}
\newcommand{\BCstateOld}{\BCstate\up{\heightOld}}

\newcommand{\validators}{\mathcal{V}}
\newcommand{\NbrVal}{\committeeSize}

\newcommand{\BCstate}{\mathcal{C}}
\newcommand{\advA}{\mathcal{A}}
\newcommand{\advB}{\mathcal{B}}

\newcommand{\client}{C}

\title{Practical Light Clients for Committee-Based Blockchains}

\author{
\IEEEauthorblockN{Frederik Armknecht}
\IEEEauthorblockA{Universität Mannheim\\
    \href{mailto:armknecht@uni-mannheim.de}{armknecht@uni-mannheim.de}}
\and
\IEEEauthorblockN{Ghassan Karame}
\IEEEauthorblockA{Ruhr-Universität Bochum\\
    \href{mailto:ghassan.karame@rub.de}{ghassan.karame@rub.de}}
\and
\IEEEauthorblockN{Malcom Mohamed}
\IEEEauthorblockA{Ruhr-Universität Bochum\\
    \href{mailto:malcom.mohamed@rub.de}{malcom.mohamed@rub.de}}
\and
\IEEEauthorblockN{Christiane Weis}
\IEEEauthorblockA{NEC Laboratories Europe\\
    \href{mailto:christiane.weis@neclab.eu}{christiane.weis@neclab.eu}}}

\begin{document}
\pagestyle{plain}
\maketitle

\begin{abstract}
Light clients are gaining increasing attention in the literature since they obviate the need for users to set up dedicated blockchain full nodes.
While the literature features a number of light client instantiations, most light client protocols optimize for long offline phases and implicitly assume that the block headers to be verified are signed by highly dynamic validators.

In this paper, we show that (i) most light clients are rarely offline for more than a week, and (ii) validators are unlikely to drastically change in most permissioned blockchains and in a number of permissionless blockchains, such as Cosmos and Polkadot.
Motivated by these findings, we propose a novel practical system that optimizes for such realistic assumptions and achieves minimal communication and computational costs for light clients when compared to existing protocols.
By means of a prototype implementation of our solution, we show that our protocol achieves a reduction by up to $90$ and $40000\times$ (respectively) in end-to-end latency and up to $1000$ and $10000\times$ (respectively) smaller proof size when compared to two state-of-the-art light client instantiations from the literature.
\end{abstract}

\section{Introduction}

\emph{Light clients} are programs that verify a digest of the blockchain's state with minimal storage, communication, and computation overhead. These clients are gaining increasing attention in the literature since they obviate the need for users to set up dedicated blockchain full nodes, store large amounts of data, and process all received information~\cite{sokLightClients}.

The literature features a number of proposals for light client implementations.
For instance, the Bitcoin community provides the BitcoinJ~\cite{bitcoinj}, PicoCoin~\cite{picocoin}, and Electrum~\cite{electrum} clients, implementing the so-called Simple Payment Verification (SPV).
In SPV, the light client has to verify the correctness of \emph{each} block header it receives.
This results in a linear communication complexity with respect to the number of block headers to be verified.
Other more recent light client instantiations, such as~\cite{popos,plumo,zkBridge}, move away from Bitcoin's proof of work model and into a setting where the blockchain depends on signatures from elected committees.
There, the light client proposals aim to optimize the (linear) communication complexity witnessed in the basic (\enquote{SPV-equivalent}) mode by leveraging additional trust assumptions from the system (e.g., having a direct communication link with an honest full node~\cite{popos}) or by incurring additional computational overhead on the full nodes \cite{plumo,zkBridge}.
As far as we are aware, Agrawal et al.'s PoPoS protocol \cite{popos} emerges as one of the most communication-efficient light client instantiations in the proof of stake setting requiring $O(\log{m})$ messages to verify $m$ consecutive block headers.

While these results greatly improve over the basic linear mode, we observe that most recent light client systems are designed with long offline phases in mind---trying to cater to the case where the number $m$ of consecutive block headers to be verified is very large.
In most popular cryptocurrencies, however, we observe that most light client requests are made by users who are rarely offline for more than a week. Moreover, we note that the most recent light client protocols implicitly assume that the block headers to be verified (by the light clients) are signed by continuously rotating keys.
While this assumption is true for Ethereum and Algorand (where entirely new committees are frequently sampled from a large universe of nodes), this is not the case in most permissioned blockchains, such as Quorum~\cite{quorum} or Hyperledger Fabric~\cite{fabric}, and in some permissionless committee-based blockchains such as Cosmos~\cite{cosmos} and Polkadot~\cite{polkadot}.

\subheading{Research Question and Solution.}
Motivated by these observations, we set forth in this paper to answer the following research question:
\emph{Can we design a light client protocol for committee-based blockchains that exhibits low communication and computation costs on both the light client and full nodes---without relying on additional trust assumptions?}
To address this question, we focus on realistic blockchain deployments in which (1) a light client with short offline periods must be served most often and (2) the signers of block headers do not drastically change in the common case.

First, we ground the former assumption on the current usage patterns of light clients in popular cryptocurrencies, such as Bitcoin. Namely, our experiments on the Bitcoin blockchain reveal that more than half of all light clients connections to full nodes occur within at most 24 hours.
Second, the latter assumption applies a priori to most permissioned blockchains as well as widely used permissionless blockchains.
This is because permissioned blockchains like Hyperledger Fabric and Quorum are usually specifically designed for usage with small committees and rare validator changes for efficiency.
On the other hand, our real-world evaluation of permissionless blockchains like Cosmos and Polkadot, shows that the membership changes that they exhibit are, in the common case, below a modest threshold relative to the total number ($<10\%$).

By leveraging such (realistic) assumptions, we introduce a new light protocol that achieves minimal communication and computational overhead on the light clients and the full nodes in the standard security model (e.g., without requiring additional trust assumptions).
We achieve this with the careful and novel utilization of \emph{transitive signatures} \cite{bellareNevenTransitive}, which have the useful property that long sequences of \emph{consecutive} messages (in our case, information related to the evolution of the blockchain) can be checked with verifier work independent of the number of messages (here, blockchain epochs).
To improve efficiency further, we exploit the compatibility of transitive signatures with existing multi-signature aggregation techniques (formalized in \cref{sec:transitiveMulti}), which further allows compressing signatures issued by different blockchain validators.
Overall, the light client's security is based on the signature scheme's proven security in conjunction with the majority-honesty assumption on validators---already assumed in most committee-based systems.
{{Our design is generic and is compatible with \emph{any} committee-based blockchain where the signers of
block headers do not drastically change over time.}}

\subheading{Contributions.}
In summary, we make the following contributions in this work:
\begin{description}[leftmargin=0.2cm]
\item[Gaps in existing protocols: ]We analyze the practical usage patterns of light clients and show that over $50\%$ of light client requests seen by full nodes concern offline phases of less than 2 hours. We also argue that validator sets are unlikely to change by large amounts in permissioned blockchains, and we empirically measured that if they changed at all in our observation period, then they only changed by up to $<0.6\%$ at a time in Cosmos, $0\%$ in XRP Ledger, and $<7\%$ in Polkadot. Existing light client protocols have not been designed to take such operating conditions into account (Section~\ref{sec:observations}).
    \item[Novel construct: ]We introduce the first practical light protocol that achieves $O(1)$ communication complexity and verifier work in the common case of a large static validator subset and concretely low computational load on the full nodes without additional trust assumptions.
Our design relies on the novel application of transitive signatures, which we generalize to multi-signatures for compatibility with the blockchain's distributed setting (\cref{sec:lightclient}).
\item[Prototype implementation: ]We implement a prototype of the system (which we plan to release as open-source to aid further research in this area) and evaluate its performance compared to the PoPoS protocol \cite{popos} and a design proposed for usage in Polkadot from Ciobotaru et al. \cite{cssv} (CSSV).
    We show empirically, by means of a prototype evaluation, that end-to-end latency for light client updates covering as many as $2^{16}$ epochs is over $90\times$ smaller than PoPoS and up to $41726\times$ smaller than CSSV.
At the same time, the communication complexity (constant at under 200 bytes per static-quorum period) is up to $1000\times$ lower than PoPoS and up to $10000\times$ lower than CSSV (Section~\ref{sec:evaluation}).
\end{description}

 \section{Preliminaries}

\subsection{Blockchain Basics}\label{sec:blockchains}

A \emph{blockchain data structure} is an append-only list of blocks.
A block's position in the blockchain is called its height.
A block contains a header and a list of transactions.
A transaction, abstractly, is an input command to a specified state machine.
A (committee-based) \emph{blockchain system} is a replicated state machine, maintained by a set of so-called \emph{validators}, using a blockchain data structure and a (committee-based) consensus protocol.
Consensus here specifically refers to Byzantine fault-tolerant atomic broadcast \cite{pbft,hotstuff}.
At all times, an honest (super)majority assumption is placed on the validators.

\subheading{Validator Reconfiguration, Epochs.}
The set of validator nodes may change over time.
Such changes are also called reconfigurations. Blockchain systems define an epoch length as a number of blocks such that reconfigurations are only allowed at block heights which are multiples of the epoch length.
In case of reconfiguration, a block with a special header field is output to announce the validator change.
These end-of-epoch blocks are also signed by the (old) validators before the reconfiguration takes effect starting in the following epoch.
Thus, the honest majority assumption on any particular set of validators is temporary---it only holds for a specific epoch.

\subheading{Committee-Based Blockchains.}
We speak of a \emph{committee-based} blockchain when the validators mutually know each other's identities. This is in contrast to, say, Bitcoin's protocol that supports an unknown set of nodes.
However, this covers those \emph{proof-of-stake} systems that select the participating nodes in a standard Byzantine consensus algorithm based on staked cryptocurrency.
Such systems are widely used and include Ethereum and Polkadot (that is, their \enquote{finalized prefixes} \cite{gasper}) and the Cosmos ecosystem based on Tendermint \cite{tendermint}.
Committe-based blockchains also include \emph{permissioned} blockchains like Hyperledger Fabric \cite{fabric}, where the assignment (and reconfigurations) of the set of validator nodes is decided by an external trusted party or procedure. In this paper, we treat such permissioned blockchains and proof-of-stake blockchains in parallel since both are committee-based.

\subheading{Light Clients, Block Verification.}
All transactions in blocks from the initial block until a given block uniquely determine the blockchain system's \emph{state} at the given block.
\emph{Light clients} are programs that verify a digest of the blockchain's state with minimal storage and computation overhead.
Using the verified digest, they might then subsequently verify, say, parts of the state itself or verify that a particular transaction is included in a given block.
It is standard that each block header includes the digest of the system state at that block.
Thus, the light client's task is verifying the block headers.
A block header is considered valid if and only if there exist valid digital signatures over the block header from a sufficient number $\quorumSize$ of validators from the current validator set. 

\newcommand{\setupBox}[2]{\fbox{\begin{minipage}[t]{4.4cm}\textbf{#1}\\#2\end{minipage}}}
\begin{figure*}[!t]
\begin{center}
\begin{tikzpicture}
    \node (V) at (0,0) {\setupBox{Validators (/committee)}{run blockchain consensus to decide new blocks.}};
    \node (F) at (6,0) {\setupBox{Full nodes}{receive blockchain data in real time; forward messages.}};
    \node (L) at (12,0) {\setupBox{Light clients}{query for blockchain data sporadically; resource constrained.}};
    
    \draw[<->] (V) -- (F) ;
    \draw[<->] (F) -- ([yshift=2\baselineskip]L) ;
\end{tikzpicture}
\end{center}
\caption{Overview of our setup (see \cref{sec:blockchains}).
\label{fig:setup}
}
\end{figure*}
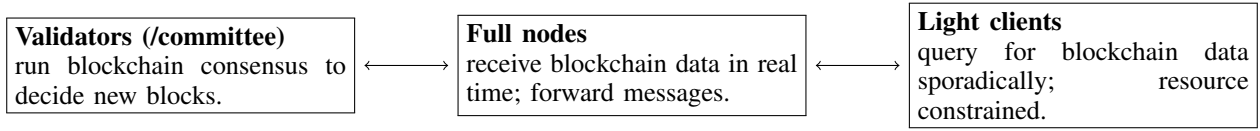 To receive updates about the blockchain, light clients do \emph{not} directly communicate with validators.
Rather, information about the blockchain's state is forwarded to them by \emph{full nodes}.
Full nodes are highly available nodes that receive blocks from validators (or, via peer-to-peer messaging, from each other) and verify and store them.
Unlike full nodes, light clients are generally not always online---they query full nodes on demand.
Figure \cref{fig:setup} illustrates the communication links in this setup.

For further context about light clients we refer to \cite{sokLightClients}.

\subsection{Cryptographic Building Blocks}\label{sec:crypto}

We recall that standard digital signature schemes are defined by three algorithms \KeyGen, \Sign, \Verify{} and are considered secure if they guarantee existential unforgeability under chosen message attacks (euf-cma).
Aggregate signatures additionally have an algorithm \Aggregate{} to combine signatures, and \Verify{} is modified to be able to check a single signature for multiple messages and multiple sets of public keys.

\subheading{BLS Signatures.}
Boneh, Lynn, and Shacham's aggregate signature scheme (BLS) \cite{bls04} is defined over three elliptic curve groups $\GG_1$, $\GG_2$, $\GG_T$ of the same large prime order $p$, generators $g\in\GG_1, \hat{g}\in\GG_2$, a hash function $\groupHash : \{0,1\}^*\rightarrow\GG_2$ and a pairing function, that is, a bilinear map $e: \GG_1 \times \GG_2 \rightarrow \GG_T$.
\KeyGen{} samples $\sk$ from $\ZZ_p$ and sets $\pk = g^{\sk}$, \Sign{} on message $x$ outputs $\sig = \groupHash(x)^{\sk}$ and $\Aggregate(\sig_1, \dots, \sig_n)$ computes $\prod_{i \in [n]} \sig_i$.
$\Verify((\PK_{1},\allowbreak{} \dots,\allowbreak{} \PK_{n}),\allowbreak{} (x_1,\allowbreak{} \dots,\allowbreak{} x_n),\allowbreak{} \sig)$ checks $e(g,\allowbreak{} \sig) = \prod_{i \in [n]} e(\prod_{\pk \in \PK_{i}} \pk,\allowbreak{} \groupHash(x_i))$.

\subheading{Transitive Signatures.}
Transitive signatures \cite{transitiveSignatures}, intuitively, are used to sign edges of graphs and allow aggregating signatures from a single signer along paths.
That is, a signature on $(x,y)$ and a signature on $(y,z)$ can be combined into a valid signature for $(x,z)$, which can be verified without regard to $y$.
The GapTS-2 scheme due to Bellare and Neven \cite[Section 5.C]{bellareNevenTransitive} has the same \KeyGen{} and \Aggregate{} procedures as BLS.
However, \Sign{} computes signatures differently, namely as $\sig = (\groupHash(y)/\groupHash(x))^{\sk}$, while \Verify{} on input $(x,z)$ checks $e(g, \sig) = e(\pk, \groupHash(z)/\groupHash(x))$.
Notably, verification complexity is independent of the path length.

\subsection{Notations}
In the sequel, we denote by $[n]$ the set $\{1, \dots, n\}$ and $[n, m]$ denotes $\{n, \dots, m\}$.
In this paper, objects related to the $i$th block height or epoch are often indexed as, for example, $x\up{i}$.

We further formalize the description of light clients given in \cref{sec:blockchains}.
We denote by $\committeeSize$ the (maximum) validator committee size and write the set of validators in epoch $i$ as $\validators\up{i}$.
Regarding light client updates, we use the variable $\missingBlocks$ for the number of epochs boundaries---that is, possible validator set changes---between the client's starting block height and the most current height.

 \section{Related Work}\label{sec:related}

\newcommand{\secondColumn}{0.4cm}
\newcommand{\thirdColumn}{0.9cm}
\newcommand{\fourthColumn}{1.45cm}
\newcommand{\fifthColumn}{1.45cm}
\newcommand{\sixthColumn}{2.65cm}
\newcommand{\seventhColumn}{0.8cm}
\newcommand{\eighthColumn}{1.4cm}
\begin{table*}[t]
    \centering
    \caption{
        Comparison of light client protocols.
        SV stands for baseline sequential verification.
The best case for PoPoS is that all queried full nodes agree.
        The best case for our scheme is the existence of a large static validator subset throughout all $\missingBlocks$ epochs.
        $T$ is a parameter for trust-based and proof-based skipping protocols.
        $\missingBlocks$ is the number of epoch blocks the light client missed since last being online.
        $\committeeSize$ is the (maximum) size of a validator committee.
        \label{tab:related}
    }
\begin{tabularx}{\textwidth}{Xp{\secondColumn}p{\thirdColumn}p{\fourthColumn}p{\fifthColumn}p{\sixthColumn}p{\seventhColumn}p{\eighthColumn}}
        \toprule
                                                                                                             & SV                             & PoPoS \cite{popos}                    & Tendermint \cite{tendermintLC}   & Algorand \cite{algoStateProofs}        & Plumo \cite{plumo}, zkBridge (batched) \cite{zkBridge} & CSSV \cite{cssv} & Our solution, Sec. \ref{sec:lightclient} \\
        \midrule
        \parbox{\linewidth}{\hangindent1em\hangafter1 Best-case communication overhead ($\Omega$)}           & --                             & $1$                                   & --                               & --                                     & --                                                     & --               & $1$ \\
        \parbox{\linewidth}{\hangindent1em\hangafter1 Worst-case communication overhead ($O$)}               & $\missingBlocks\committeeSize$ & $\log(\missingBlocks)$                & $\missingBlocks\committeeSize/T$ & $\missingBlocks\log(\committeeSize)/T$ & $\missingBlocks/T$                                     & $\missingBlocks$ & $\missingBlocks$ \\
\parbox{\linewidth}{\hangindent1em\hangafter1 Secure if all full nodes are malicious}                  & \fullCircle                    & \emptyCircle                          & \fullCircle                      & \fullCircle                            & \fullCircle                                            & \fullCircle      & \fullCircle \\
\parbox{\linewidth}{\hangindent1em\hangafter1 Epochs for which validators must \emph{stay} majority-honest}   & 1                              & 1                                     & $T$                              & $T$                                    & 1                                                      & 1                & 1 \\
        \parbox{\linewidth}{\hangindent1em\hangafter1 Low computational load on full nodes / provers}        & \fullCircle                    & \fullCircle                           & \fullCircle                      & \fullCircle                            & \emptyCircle                                           & \halfCircle      & \fullCircle \\
\bottomrule
    \end{tabularx}
\end{table*} 
We now overview existing light client approaches for committee-based blockchains.
\cref{tab:related} summarizes the core strengths and drawbacks.
We remark that none of the existing protocols cater to the best-case scenario where the validator set is relatively static.

\subheading{Sequential Verification (SV).}
We refer to the following standard baseline approach to light client verification as sequential verification (SV).
At the start, the light client is initialized with the first committee's public keys.
To update to the current block header within the same epoch, the full node simply sends the block header and enough validators' signatures to show that the block is valid.
To update across epochs, the full node sends all end-of-epoch block headers with validator set changes.
The light client then checks them in sequence, verifying the first one with the known keys and verifying each subsequent one with the public keys taken from the last one.
Finally, the present epoch's public keys are used to verify the most recent header.
{ Sequential verification schemes are reminiscent of Bitcoin's SPV, though in the committee-based setting. SV schemes are supported, for example, by Ethereum (implemented in the Helios client \cite{helios}) and Polkadot (called \enquote{warp sync}).}

\subheading{PoPoS.}
The proofs of proof-of-stake (PoPoS) of \cite{popos} optimizes the SV approach as follows: instead of sequentially checking the validators' signatures, the light client connects to multiple full nodes, at least one of which is assumed to be honest, and informs them of its local height $\startHeight$.
Then, each full node directly sends the most recent header.
In the optimistic case, they all agree, and the update finishes immediately.
Otherwise, if there is disagreement, the light client starts an interactive bisection protocol to identify the honest full node.
Downsides to this approach are 1) that without an honest full node, the light client might accept a false state, 2) that even if there is an honest full node, a network-level adversary might prevent the light client from receiving the honest node's responses and 3) the added latency of $O(\log(\missingBlocks))$ round-trip times due to the protocol's interactivity.

\subheading{Skipping Optimizations.}
Several proposals improve the SV approach via mechanisms that allow verifying multiple steps at once or \enquote{skipping ahead}.
Examples are Tendermint~\cite{tendermintLC}, Algorand~\cite[Section 8]{algoStateProofs} and Plumo~\cite{plumo}.
Each skipping protocol introduces a parameter, which we uniformly denote as $T$, for the number of steps that can be skipped.

The Tendermint design \cite{tendermintLC} is based on the assumption that honest validators, even after removal from the committee, must continue to sign the next $T$ end-of-epoch blocks and stay honest during this so-called trusting period.
Thus, the Tendermint light client starting at height $\startHeight$ stores, in addition to the validators' public keys, the height $h'<\startHeight$ when this set of validators was appointed\footnote{We simplify for ease of explanation. In \cite{tendermintLC}, the client stores the inauguration height for each validator individually rather than for the committee as a whole.}.
With this information, it is possible to quickly verify the epoch block just before $h'+T$ to obtain a new set of keys.
Iterating, it is only necessary to check every $T$th block to reach the latest epoch, at which point the latest block is verified directly.

Algorand's light client design~\cite[Section 8]{algoStateProofs} follows a similar principle.
There, every $T$th block contains a hash\footnote{Let it suffice to say that the specific use of (Merkle tree) hashes serves as a means to reduce the communication complexity \emph{per update} from $O(\committeeSize)$ to $O(\log(\committeeSize))$.} of the list of public keys of a subset of the current validators.
This subset is responsible for signing the next $T$th block, that is, the subset is implicitly trusted for $T$ blocks in the future.
This kind of trust assumption needed for Tendermint and Algorand might be backed by proof of stake and the threat of financial punishments for detectable misbehavior.

\subheading{Cryptographic Proofs.}
Plumo \cite{plumo} proposes the reliance on a cryptographic proof for $T$ steps of sequential verification, implementing a proof-based version of \enquote{skipping ahead}.
This approach is secure even if the trusting period is short, e.g., only one epoch. However, it comes with the drawbacks of 1) significant computational cost for proof creation and 2) added latency, as a new proof can only be created every $T$ blocks.

Finally, zkBridge~\cite{zkBridge} and CSSV~\cite{cssv} aim to reduce the per-step cost of sequential verification.
Both proposals do not introduce additional trusted parties.
In short, each step of the sequential verification is replaced with a cryptographic proof that the full node sends to the light client in place of the validators' $O(\committeeSize)$ signatures.
Specifically, zkBridge proposes using a general-purpose cryptographic proof, whereas CSSV presents an apparently more efficient scheme that is fine-tuned to systems where validators use BLS signatures over the elliptic curve BLS12-377.
In the scenario where a light client returns after an offline phase, $\missingBlocks$ proofs are needed.
As $\missingBlocks$ increases, this still has significant overall computational costs for proof generation and verification\footnote{\cite{zkBridge} also supports batching, in which case only $\missingBlocks/T$ proofs are needed (as in the \enquote{skipping} optimizations).
As per the evaluation presented in~\cite{zkBridge}, generating a proof that covers $T = 32$ blocks at once with $\committeeSize = 128$ ECDSA signatures per block still takes 18 seconds on powerful hardware.}.
 \section{Problem Statement}

In this section, we motivate our research problem and outline our security model.

\subsection{Motivation}\label{sec:observations}

We formulate our research problem based on the following observations of current real-world deployments.

\begin{observation}\label{obs:m}
Existing light client designs often \emph{choose between} light client efficiency (1) after long offline phases ($\missingBlocks \rightarrow \infty$) or (2) for frequent small-$\missingBlocks$ updates.
Empirically, in the most popular cryptocurrency with widely used light clients, the majority of requests \emph{that full nodes receive} correspond to small values of $\missingBlocks$ equivalent to under 2 hours of missed blocks.
\end{observation}

\begin{figure}
    \centering
\includegraphics[width=0.87\linewidth]{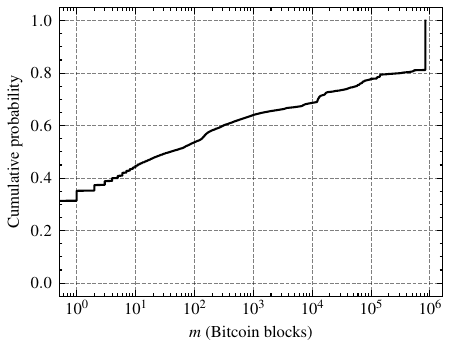}
\caption{Empirical cumulative distribution function of measurements of $m$ in Bitcoin. {We include the data used to create this plot in Table~\ref{tab:cum} in the Appendix.} \label{fig:m}}
\end{figure}

 \begin{figure*}[tbp]
    \centering
    \begin{subfigure}{0.32\textwidth}
        \includegraphics[width=\linewidth]{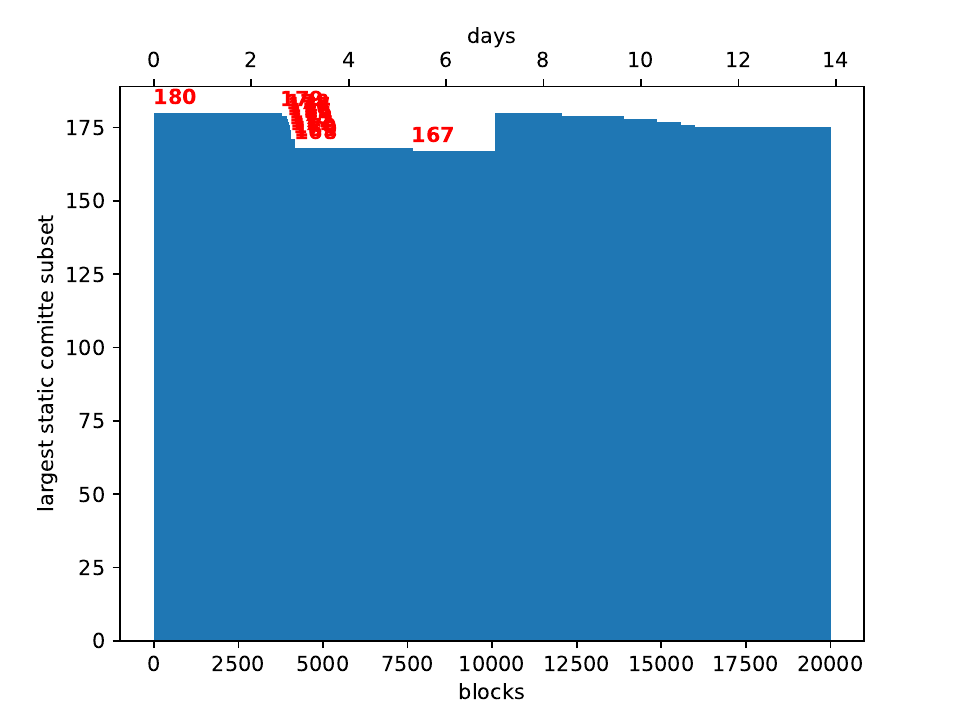}
         \caption{Cosmos, epoch length 1 block.}
\end{subfigure}
\begin{subfigure}{0.32\textwidth}
        \includegraphics[width=\linewidth]{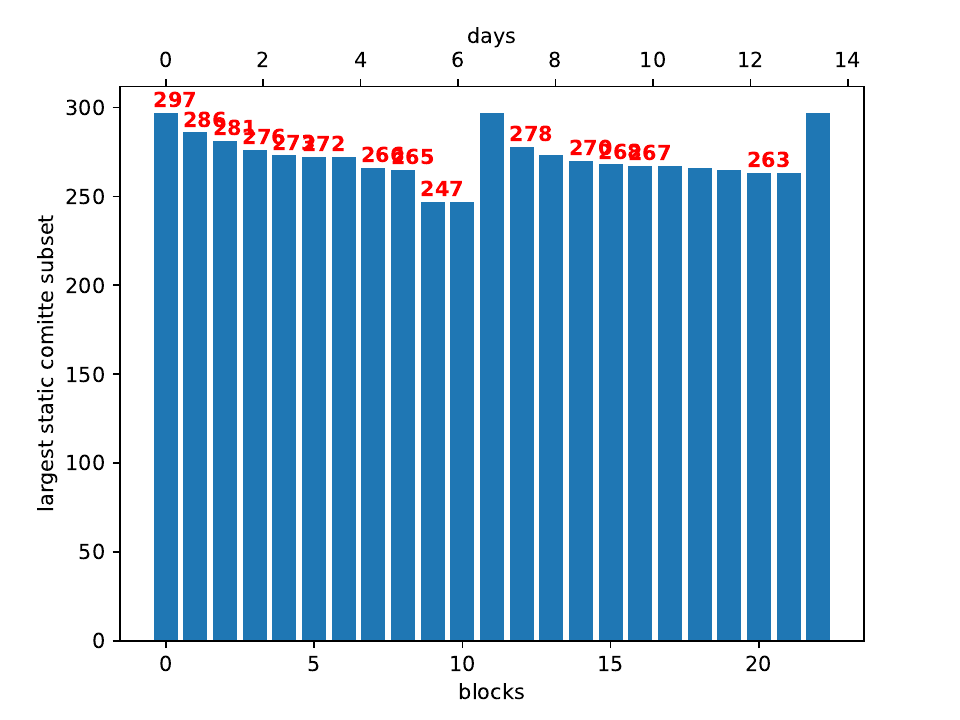}
         \caption{Polkadot, epoch length 1 day.}
\end{subfigure}
    \begin{subfigure}{0.33\textwidth}
        \includegraphics[width=\linewidth]{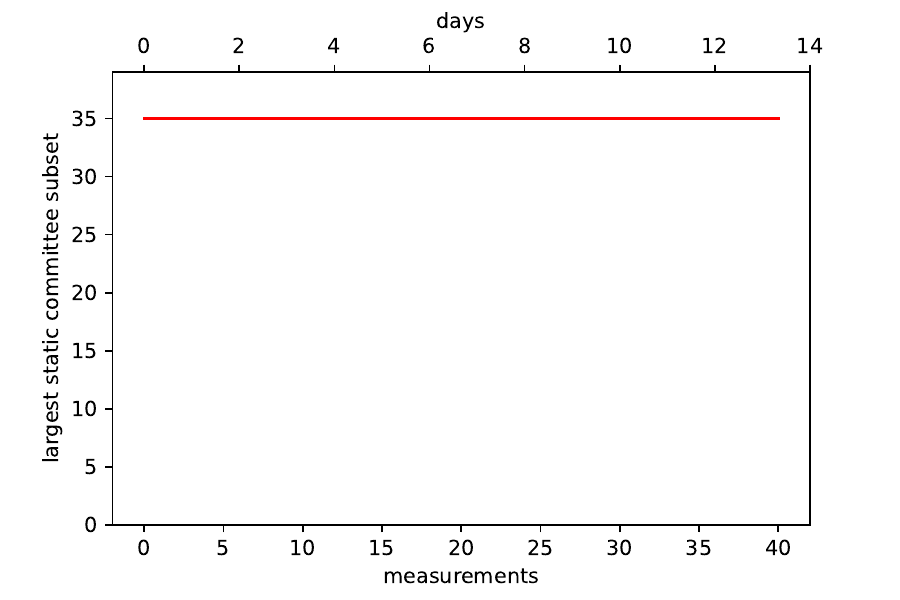}
         \caption{XRP Ledger, epoch length unspecified.}
\end{subfigure}
    \caption{Maximum size of static validator subsets over a two-week period. Absolute values are indicated in red. {We include the data used to create this plot in Table~\ref{tab:quorumsize} in the Appendix.}
    \label{fig:committees}}
\end{figure*}

First, we note that the primary use cases of current light client designs differ.
Some designs, such as Algorand and Plumo, efficiently allow clients to sync after long offline phases.
As such, they particularly cater to initial bootstrapping, where a light client only has the initial blockchain state and must synchronize up to the present. However, they provide no benefit for \enquote{normal daily usage}, that is, for clients updating from starting heights less than $T$ behind the present, since proofs are only created after every $T$th epoch.
On the other hand, approaches such as zkBridge and CSSV purely focus on the case $\missingBlocks = 1$.
They succeed in reducing the effect of $\committeeSize$, which is desirable in the small-$\missingBlocks$ regime, but suffer linear complexity when considering longer updates.
(This refers to zkBridge without batching---with batching, the previous consideration applies since one would then have to wait until a batch is full before experiencing a benefit.)

The second part of \cref{obs:m} stems from the following experiment which we conducted on Bitcoin to measure realistic absolute values of $\missingBlocks$ that full nodes encounter.
We opted to use Bitcoin since it is the most popular blockchain and was the first blockchain to support lightweight clients.
{{That is, even though Bitcoin is \emph{not} a committee-based blockchain, we do not expect drastically different lightweight client usage patterns compared to (committee-based) cryptocurrencies such as Cosmos.  Namely, we argue that lightweight client usage patterns are largely orthogonal of the underlying cryptocurrency and are instead driven by human behavior. For example, the way users interact with online banking applications is not expected to vary significantly between two comparable commercial banks. Furthermore, Bitcoin provides access to a much larger dataset, as users in ecosystems like Ethereum, Cosmos, or Polkadot more frequently connect to blockchains through centralized providers that act as trusted intermediaries. 
This prohibits the data gathering described next\footnote{Indeed, we attempted a similar experiment on a committee-based chain, Polkadot, but could not observe any light client request over the measurement period from our public node.}.}}
To this end, we set up a Bitcoin full node for a period of 2 weeks and measured when the incoming light clients were last online, based on their so-called \texttt{version} and \texttt{filterload} messages which, respectively, reveal their last block height and identify them as light clients.
Over the measurement period, 19458 samples were obtained.
As shown in~\cref{fig:m}, the results indicate that over a third of the light client requests queried for $\missingBlocks = 1$ block, that is, around 10 minutes in Bitcoin, and more than half asked for updates regarding $\missingBlocks = 12$, around 2 hours.

\begin{observation}\label{obs:trust}
Existing light client designs achieve improvements over the $\m \NbrVal$ communication overhead of the basic sequential verification (SV) either by exploiting additional trust assumptions or by incurring considerable computational load on the full nodes.
\end{observation}

As laid out in \cref{sec:related}, the surveyed designs all either (1) rely on honest full nodes for their security \cite{popos}, (2) require validators to remain honest for longer periods \cite{tendermintLC,algoStateProofs} or (3) introduce high computational costs for cryptographic proof generation \cite{plumo,zkBridge,cssv}.
When trying to use the designs in systems other than those for which they were originally developed, it might be difficult to determine whether the required trust assumptions are justified in that system or whether the proving costs are acceptable.
For example, PoPoS \cite{popos} works under the assumption that a light client has a direct communication link to at least one honest full node.
This assumption might be acceptable in some situations, but one might conceive of cases where it is not desirable.
If, say, a mobile cryptocurrency wallet application was based on PoPoS, a network-level adversary might violate the assumption by dropping the honest node's messages and might cause the light client to accept a false blockchain state---the wallet would display the wrong account balance or display false transactions\footnote{As mentioned in \cite{popos}, such scenarios can be mitigated---though this might come with additional restrictions like added latency for a \enquote{dispute period} in which light clients will remain open to receiving contradicting claims.}.

Regarding proof creation costs, we remark that the choice of the underlying cryptographic schemes plays an important role.
For instance, Plumo~\cite{plumo} relies on validators using a specific signature scheme and hash function when they sign end-of-epoch blocks.
zkBridge~\cite{zkBridge} relies on a more costly EdDSA signature scheme (even when using the new cryptographic proof system that optimizes for low overhead) when compared to CSSV~\cite{cssv}.
Indeed, CSSV is fine-tuned to BLS signatures over a special elliptic curve and lets full nodes create proofs within hundreds of milliseconds---an improvement of two orders of magnitude---while the overall protocol structure is the same as zkBridge.

\begin{observation}\label{obs:committees}
When considering existing designs that do not require honest full nodes, we note that the most costly aspect of the verification logic of such approaches is handling arbitrary validator set changes.
However, we observe that in some widely used blockchains, large subsets of validators often remain across multiple epochs, and the relative size of validator changes is often less than 7\%.
\end{observation}

The first part of \cref{obs:committees} specifically regards the basic SV and the cryptography-based improvements to SV of Plumo, zkBridge, and CSSV.
We point out that in the case that validator changes are rare, a naive implementation of SV would lead to redundantly transmitting (or, at the very least, performing a verifier check over) $\committeeSize$ public keys for each epoch.
A more efficient design would avoid this and, where possible, batch verify signatures with respect to the same public keys since that could be faster.
Plumo, zkBridge, and CSSV are constructed to provide cryptographic proofs for the standard SV logic for arbitrary validator changes, but they do not seem to take such opportunities into account.

For the second part of this observation, we again point to an empirical analysis we conducted.
Over the course of two weeks, we obtained the validator sets from public blockchain node endpoints for three widely used blockchain systems: Cosmos, Polkadot, and XRP Ledger.
For each epoch, we then computed the size of the largest subset of the epoch's validators, which was also in the following epoch's validator set.
This is shown in \cref{fig:committees}.
Notice how for Cosmos and Polkadot, subsets of 167 and 247 validators are present in \emph{all} observed epoch's validator sets.
These are large majorities of the at most 180 and 300 total validators of Cosmos and Polkadot in the measured period.
The largest relative change from one epoch to the next was under $7\%$.
Further, as seen in \cref{fig:committees}, in XRP Ledger, \emph{all} validators remained in the validator set across the measurement period\footnote{In the particular case of XRP Ledger, nodes choose their own validator set. Our experiments on 5 randomly selected nodes show that nodes typically rely on a static set comprising, on average, 35 validators.}.
For permissioned systems like Hyperledger Fabric \cite{fabric} and Quorum, where nodes are controlled by a few parties, we expect validator changes to be rare (similar to XRP Ledger) since reconfigurations (changes in the validator set) tend to be very costly in permissioned settings.

\subheading{Research Question.}
We consider the aforementioned observations in conjunction.
For instance, \cref{obs:m} and \cref{obs:committees} suggest that full nodes are likely to frequently receive small-$\missingBlocks$ light client update requests---where the update heights likely fall within static validator periods.
This leads us to the following research question:

{\it Is it possible to design a light client system for committee-based blockchains that (i) has a low concrete communication cost under realistic conditions, i.e., when $\missingBlocks$ is small, and committees changes are rare and (ii) that relies on minimal trust assumptions without (iii) incurring considerable computational overhead on full nodes?}

\subsection{Threat Model}\label{sec:model}

\subheading{Security Goal.}
We say that there is one correct global state $\BCstateOld$ of the underlying blockchain at each height $\heightOld$ and one correct validator set $\validators\up{i}$ in each epoch $i$. 
Then, a light client at blockchain height $\startHeight$ in epoch $\startEpoch$ stores in its local memory data $\sOld$ which, importantly, includes the current validators' ($\validators\up{\startEpoch}$) public keys (i.e., $\pk_v$ for $v \in \validators\up{\startEpoch}$) and a short digest of the blockchain's global state. In the $\Update$ protocol, the client interacts with a set of full nodes $\FNClient$, after which the client determines its new local state $\sNew$ corresponding to height $\currentHeight$ in epoch $\currentEpoch$.
{{In brief, our main security goal is to ensure that the light client can be sure with overwhelming probability that the resulting new local state $\sNew$ is indeed the correct global state.}}

Formally, let $\corresponds$ be the function that given a global blockchain state $\BCstateOld$ returns the (unique) corresponding correct light client state $\sOld$.
The following definition describes an update procedure as $(\sNew,\allowbreak{} \pi)\allowbreak{} \gets\allowbreak{} \Update_{\FNClient}(\sOld)$, meaning a query is sent to full nodes in $\FNClient$, each of sends (1) a response from which light client derives its new state $\sNew$ and (2) a proof $\pi$.
The definition guarantees that a client starting from the correct state at height $\startHeight$ either gets the light client state that corresponds to the global blockchain state of the new height $\currentHeight$ or the verification of the received state fails (except with negligible probability).

\begin{definition}\label{def:UpdSecurity}
We say that a light client scheme ($\Update,\allowbreak{} \VerifyState$) is secure if for all probabilistic polynomial-time adversaries $\advA$, all choices of global blockchain states $\BCstate\up{\startHeight}, \BCstate\up{\currentHeight}$, total blockchain lengths $\currentHeight$ and all heights $\startHeight \leq \currentHeight$, it holds that the success probability defined by 

\[\bigPr{
    \VerifyState(\sOld,\lcState, \\
    \pi) = 1 \wedge \lcState \neq \sNew
}{
    \sNew \gets \corresponds(\BCstate\up{\currentHeight})\\
    \sOld \gets \corresponds(\BCstate\up{\startHeight})\\
    (\lcState,\pi) \gets \Update_{\advA}(\sOld)
}\]
is negligible in the security parameter $\secparam$, where $\Update_{\advA}$ is the update protocol run with the adversary $\advA(\secparam, \startHeight,\currentHeight)$ choosing the behavior of all counterparties.
\end{definition}

\subheading{Attacks and Attackers.}\label{sec:attackers}
{{To motivate our attacker model, we show how a light client could be tricked into accepting a wrong state $sNew$. Recall that the intended information flow is that a full node $\FullNode$ collects information about the current state from the set of validators $\validators$ and uses this information to generate the response to the light client $\client$, i.e., we have
$\client \leftrightarrow \FullNode\leftarrow \validators.$ 
Here, we do not make any explicit assumptions with respect to guaranteed message delivery between any parties.
In particular, since an unreliable asynchronous (or adversarial) network environment might lead to the full nodes' messages to light clients getting dropped, unlike other works such as \cite{popos}, we will not assume that there always exists an honestly responding full node to every light client request.

On the contrary, we cannot exclude that $\FullNode$ and/or some of the validators in $\validators$ arbitrarily misbehave. In principle, if the client $\client$ gets wrong information, i.e., information that leads to a wrong state, this can only happen if one of the following situations do occur:
\begin{enumerate}
    \item All validators are honest, and the dishonest full node $\FullNode$ provides to $\client$ an outdated state, i.e., the state was correct at some previous point in time. \label{attack:OutdatedState}
    \item All validators are honest, and the dishonest full node $\FullNode$ provides to $\client$ an incorrect state, i.e., the state was never the global state (neither now or in the past) \label{attack:FakedState}
    \item Some validators are dishonest, and the dishonest full node $\FullNode$ provides to $\client$ an outdated state, i.e., the state was correct at some previous point in time. \label{attack:OutdatedStateCorruptedV}
    \item Some validators are dishonest, and the dishonest full node $\FullNode$ provides to $\client$ an incorrect state, i.e., the state was never the global state (neither now or in the past) \label{attack:FakedStateCorruptedV}
\end{enumerate}

With respect to Cases (\ref{attack:OutdatedState}) and (\ref{attack:OutdatedStateCorruptedV}), there can be no formal guarantee that any received state by the light client is the \enquote{most recent} one (new blocks might have come out before the full nodes' responses are received; full nodes themselves might not have received the most recent blocks yet).
In practice, a lightweight client can connect to multiple full nodes to ensure that at least one full node is honest and has access to the latest state. Moreover, the light client may rely on metadata to detect outdated data. For example, assuming that the light client $\client$ typically knows how long it has been offline, the epoch number of the (claimed) current state could indicate whether data is outdated. 
Case (\ref{attack:FakedState}) actually requires the full node $\FullNode$ to either drop or alter inputs coming from honest validators. Recall that the client $\client$ is assumed to know the public keys of the set of validators $\validators$. This prevents the full node from tampering with any messages signed by validators.

Consequently, the attacker model boils down to Case (\ref{attack:FakedStateCorruptedV}), where a bounded subset of the validators is malicious and cooperatively tries to forge a proof for a wrong state.
}}
Here, we assume a dynamic corruption of blockchain validators.
That is, for each epoch $i$, each of the $\committeeSize$ validators $v \in \validators\up{i}$ is either in the set $\validatorsAdv\up{i}$ of corrupted validators or in the set of honest validators.
The number of corrupted validators in any epoch $i$ is bounded by $\validatorsThreshold$ (e.g., $\validatorsThreshold \leq \frac{\NbrVal-1}{3}$ in traditional Byzantine fault tolerant protocols~\cite{pbft})\footnote{The consensus protocol determines, based on $\validatorsThreshold$, the minimum quorum size $\quorumSize$, that is, the minimum number of validators which must sign a block.}: \[\forall i: \; |\validatorsAdv\up{i} \cap \validators\up{i} | \leq \validatorsThreshold.   \]
The used consensus protocol then determines, based on $\validatorsThreshold$, the minimum quorum size $\quorumSize$, that is, the minimum number of validators which must sign a block.
Corrupted parties can deviate arbitrarily from the protocols, except that corrupt validators cannot sign messages regarding previous epochs $<i$ in which they were honest. This models, e.g., the use forward-secure signatures where private keys have multiple components, one associated with each epoch, and honest validators delete the epoch-specific components at the end of the epoch.

 \section{Light Client Protocol}\label{sec:lightclient}

In this section, we introduce our light client protocol, which exhibits communication complexity asymptotically linear in the number of periods where a quorum of validators remains the same.
In our design, validators sign epoch blocks using a signature scheme that is carefully tuned to later allow untrusted full nodes to \emph{efficiently prove that a fixed validator subset remained static} to a light client without additional trust assumptions.
Last but not least, we note that our design explicitly takes into account the observations of \cref{sec:observations} and mostly optimizes for realistic deployment scenarios where the validator set is relatively static.

Before introducing our protocol, we start by describing a strawman solution and use it as a basis to introduce our full protocol later on.

\subsection{Strawman Solution}

Consider $\missingBlocks$ completed consecutive epochs in which the client has been offline.
For simplicity, we use epochs $1, \dots, \missingBlocks$ in this overview (such that the current epoch is $\missingBlocks+1$), and we focus on the interaction with one full node.
The client's goal is to verify the latest signed digest of the blockchain's state---corresponding to a height that falls within the current epoch $\missingBlocks+1$---, given that the client knows the keys belonging to the validator set $\validators\up{1}$ of epoch 1.
Recall that in the standard design, at the end of every epoch $i$, validators publish signatures over the upcoming epoch's new validator set.
For any epoch $i$, let $\sig\down{v}\up{i}$ denote this signature created by a validator $v$ at the end of epoch $i$.
We formalize the exact message being signed as \enquote{$(i, \Nextkeys\up{i})$}.
In the strawman design, $\Nextkeys\up{i}$ is, naturally, exactly the list of all public keys of the next epoch's validators,
\[
\Nextkeys\up{i} = \{\pk_v \mid v \in \validators\up{i+1}\}.
\]

\subheading{Static \NextKeys/.}
Let us assume first that there was no change at all in the validator set across all epochs in $[\missingBlocks]$. The main idea is that validators should treat the case of no change from the previous epoch differently when they output the end-of-epoch signatures.
To serve a light client update, the full node will then aggregate all signatures $\sig\down{v}\up{i}$ into one signature $\Sig$ such that the client needs only to check $\Sig$ to verify that there was indeed no validator change in the whole period.
Now, what kind of special-case signature should be used in order to allow for efficient aggregation of this kind?

As a starting point, in order to ensure that $\Sig$ is really coupled with the correct epoch period, a natural approach is that $\sig\down{v}\up{i}$ is just a signature of \enquote{$i$}, that is $\sig\down{v}\up{i}=\Sign(\sk\down{v},i)$.
With standard aggregatable BLS signatures, where aggregation simply means multiplication (see \cref{sec:crypto}), the sketched protocol would result in
\begin{equation*}
    \Sig = \left(\groupHash(1)\cdot\ldots\cdot \groupHash(\missingBlocks) \right)^{\sumKey}
\end{equation*}
where $\groupHash$ is a hash function to an elliptic curve and $\sumKey$ is the sum of the secret keys of the validators whose signatures were aggregated.
While the client would now need to validate one signature only, the verification would still require computing a product of $\missingBlocks$ different values.
Thus, the light client's effort grows linearly with $\missingBlocks$.
This is already acceptable when $\missingBlocks$ is small---a frequent occurrence, as seen in \cref{obs:m}---, but we can still do better and \emph{circumvent} the \enquote{low $\missingBlocks$ or high $\missingBlocks$} tradeoff that existing designs make (\cref{obs:m}).

\subheading{Reducing Verifier Work.}
To overcome this effort, we draw on so-called transitive signatures and construct signatures according to the scheme GapTS-2 of \cite{bellareNevenTransitive} (\cref{sec:crypto}).
That is, signatures are quotients such that $\groupHash(i)$ values cancel out when signatures are multiplied together.
Only the start and end points remain.
Concretely, using
\begin{equation*}
    \sig\down{v}\up{i} = \left(\groupHash(i)/\groupHash(i-1)\right)^{\sk\down{v}}
\end{equation*}
creates the following upon aggregation:
\begin{align*}\label{eq:aggregate}
    \Sig &= \prod_{i \in [m]} \prod_{v} \sig\down{v}\up{i} \\
         &= \prod_{i \in [m]} \left(\groupHash(i)/\groupHash(i-1)\right)^{\sumKey} \\
         &= \left(\groupHash(\missingBlocks)/\groupHash(0)\right)^{\sumKey},
\end{align*}
for which verification is $O(1)$ as desired.
 
\subheading{Dynamic \NextKeys/, Break Points.}
Next, we consider the case that the validator set changed.
Formally, we have some index $\bp \in [m]$ such that $\Nextkeys\up{0} = \ldots = \Nextkeys\up{\bp-1}$ but $\Nextkeys\up{\bp-1}\allowbreak{} \neq \Nextkeys\up{\bp}$.
We call such an index a \emph{break point} (BP).
We extend our approach above as follows: for any epoch $i$, any validator $v$ produces the signature
\begin{equation}\label{eq:signing}
    \sig\down{v}\up{i} = \left\{\begin{array}{ll}
         \left(\groupHash(i, \Nextkeys\up{i})/\groupHash(i-1)\right)^{\sk_v} & \text{if } i \text{ is a BP} \\
         \left(\groupHash(i)/\groupHash(i-1)\right)^{\sk_v} & \text{else.} 
    \end{array}\right.
\end{equation}
The role of $\Nextkeys\up{i}$ here is twofold:
\begin{enumerate}
    \item indicating that a change in the set of validators takes place, as well as
    \item confirming the validity of the new validator set.
\end{enumerate}
We now define $\bp_1,\ldots,\bp_{\ell-1}$ as all break points in the range $1,\ldots,\missingBlocks-1$ and furthermore $\bp_0 = 0$ and $\bp_\ell = \missingBlocks$.
We divide the given period into subperiods $[\bp_0+1,\bp_1]$, \ldots, $[\bp_{\ell-1}+1,\bp_\ell]$.
For each subperiod $k\in [\ell]$, the full node generates the aggregated signature 
\begin{align*}
    \Sig_{k} &= \prod_{i\in [\bp_{k-1}+1, \bp_k]} \prod_{v} \sig\down{v}\up{i}.
\end{align*}
As the intermediate quotients cancel out as before, it holds that 
\begin{equation*}
    \Sig_{k} = \left(H(\bp_k, \Nextkeys\up{\bp_k})/\groupHash(\bp_{k-1}) \right)^{\sumKey_{k}}
\end{equation*}
where $\sumKey_{k}$ is the sum of secret keys of a subset of validators of this epoch range.
For the last subperiod, we have
\begin{equation*}
    \Sig_{\ell} = \left\{\begin{array}{ll}
        \left(\groupHash(\missingBlocks, \Nextkeys\up{\missingBlocks})/\groupHash(\bp_{\ell-1})\right)^{\sumKey_\ell} & \text{if } \missingBlocks \text{ is a BP} \\
        \left(\groupHash(\missingBlocks)/\groupHash(\bp_{\ell-1})\right)^{\sumKey_\ell} & \text{else.}
    \end{array}\right. 
\end{equation*}
Finally, the full node aggregates all signatures $\Sig_k$ and sends it to the client along with the break point indices and all the sets of public keys $\Nextkeys\up{\bp_1}, \dots, \Nextkeys\up{\bp_\ell}$.
Verification is clearly $O(\ell)$.

\subsection{Final Protocol}

\newcommand{\fullNodeSpace}{7.8cm}
\newcommand{\lightClientSpace}{7.8cm}
\newcommand{\arrowLength}{1cm}

\begin{figure*}[!t]
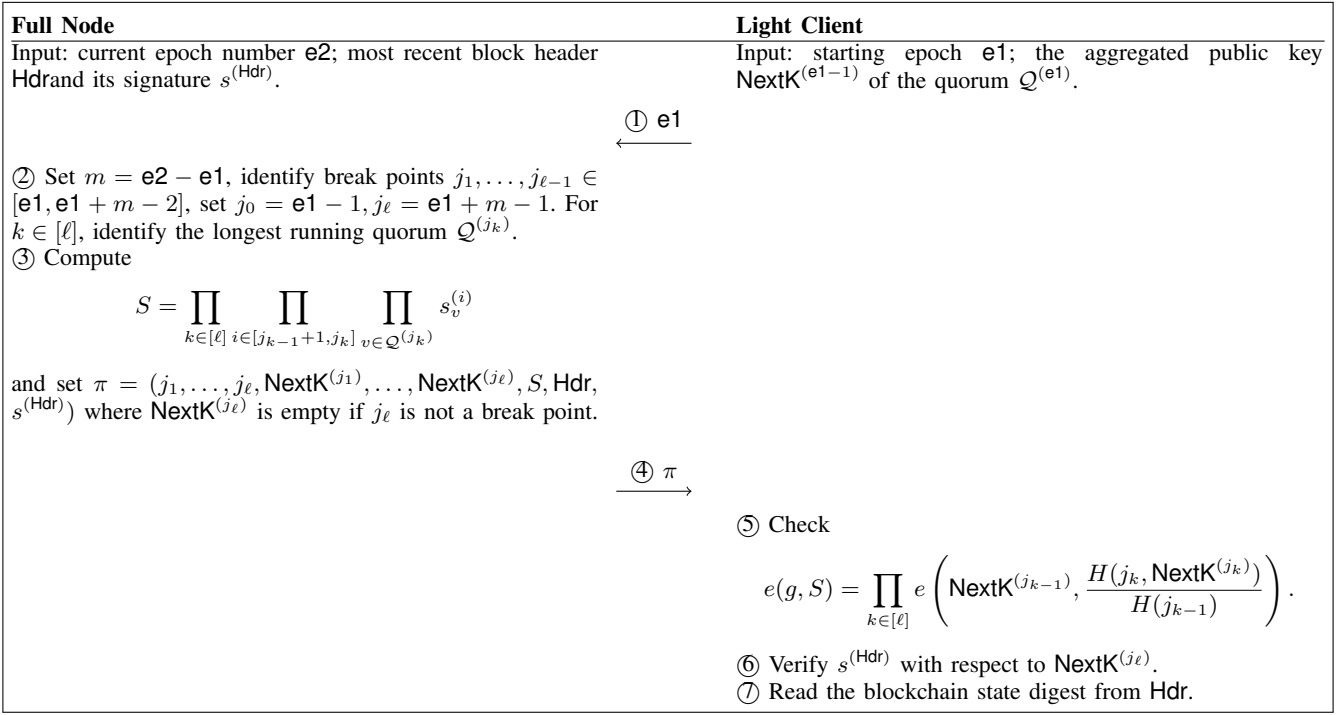

\begin{pcvstack}[boxed, center]\procedure[colspace=0cm]{}{\\[-2.5ex]
\textbf{Full Node} \<\< \textbf{Light Client} \\[][\hline]
    \begin{minipage}[t]{\fullNodeSpace}
        Input:
        current epoch number $\currentEpoch$;
most recent block header \newestHeader and
        its signature $\sig\up{\newestHeader}$.
    \end{minipage} \<\< \begin{minipage}[t]{\lightClientSpace}
        Input: starting epoch $\startEpoch$; the aggregated public key $\Nextkeys\up{\startEpoch-1}$ of the quorum $\quorum\up{\startEpoch}$.
    \end{minipage}\\
    \< \sendmessageleft*[\arrowLength]{\text{\step{1} } \startEpoch} \\
    \begin{minipage}{\fullNodeSpace}
        \step{2} Set $\missingBlocks = \currentEpoch - \startEpoch$, identify break points $\bp_1, \dots, \bp_{\ell-1} \in [\startEpoch, \startEpoch+\missingBlocks-2]$, set $\bp_0 = \startEpoch-1, \bp_\ell=\startEpoch+\missingBlocks-1$.
        For $k \in [\ell]$, identify the longest running quorum $\quorum\up{\bp_k}$. \\ \step{3} Compute
        \[
            \Sig = \prod_{k \in [\ell]} \prod_{i \in [\bp_{k-1}+1, \bp_k]} \prod_{v \in \quorum\up{\bp_k}} \sig\down{v}\up{i}
        \]
        and set $\pi = (\bp_1, \dots, \bp_{\ell}, \Nextkeys\up{\bp_1}, \dots, \Nextkeys\up{\bp_{\ell}},\allowbreak \Sig,\allowbreak \newestHeader,\allowbreak \sig\up{\newestHeader})$ where $\Nextkeys\up{\bp_{\ell}}$ is empty if $\bp_{\ell}$ is not a break point. \\
\end{minipage} \\
    \< \sendmessageright{length=\arrowLength, top={\step{4} $\pi$}} \\
    \<\< \begin{minipage}{\lightClientSpace}
\step{5} Check
    \[
    e(g, \Sig) = \prod_{k \in [\ell]} e\left(\Nextkeys\up{\bp_{k-1}}, \frac{\groupHash(\bp_k, \Nextkeys\up{\bp_k})}{\groupHash(\bp_{k-1})}\right).
    \]
    \step{6} Verify $\sig\up{\newestHeader}$ with respect to $\NextKeys/\up{\bp_{\ell}}$. \\
    \step{7} Read the blockchain state digest from \newestHeader.
    \end{minipage}
}\end{pcvstack}
\caption{
    Light client protocol for an update from epoch $\startEpoch$ to the current epoch $\currentEpoch$.
\label{fig:pseudocode}
    \vspace{-0.5 em}
}
\end{figure*} \subheading{Longest Running Quorum.}
The strawman design gives a protocol where the light client is $O(1)$ as long as the full validator set remains exactly the same.
We now move to the setting where only a static quorum is required.
This allows us to cater to minor changes in the validator set, like in Cosmos and Polkadot (see Figure~\ref{fig:committees}).
More specifically, we assume that validators keep track of the churn in the system (which is typical in most consensus deployments) and identify a subset of nodes that are most likely to remain static over time.
Suppose that each validator in epoch $i$ keeps a record of the number of epochs for which every (other) validator has already been in the set.
Each validator can then determine the unique subset $\quorum\up{i} \subset \validators\up{i}$ that has constituted a quorum for the most consecutive epochs so far (compared to all other subsets).
Ties can be broken using any deterministic rule to ensure that this \emph{longest running quorum} is unique\footnote{In full generality, the protocol works for any unambiguous definition of the longest running quorum.}.

\subheading{From $\validators$ to $\quorum$.}
With this, the only change the full protocol makes to the strawman is that we re-define $\Nextkeys\up{i}$ as a single value, the next longest running quorum's combined public key: \[
\Nextkeys\up{i} = \prod_{v \in \quorum\up{i+1}} \pk_v.
\]
Consequently, break points are also based on $\quorum\up{i}$ rather than $\validators\up{i}$.
This guarantees that valid signatures of the type $\groupHash(i)/\groupHash(i-1)$ issued by all validators in $\quorum\up{i}$ exist as long as the set $\quorum\up{i}$ remains a quorum.
Furthermore, such signatures \emph{only} exist as long as $\quorum\up{i}$ remains a quorum due to the honest majority assumption, which ensures security.
All of these signatures can be aggregated and verified in $O(1)$, making the protocol as efficient as desired.

With no further modifications, we end up with the simple 2-move protocol in \cref{fig:pseudocode}, which we briefly walk through.
First, the light client sends the starting epoch $\startEpoch$ (\step{1}).
Next, the full node must find the break points $\bp_1, \dots, \bp_{\ell-1}$ between $\startEpoch$ and the current epoch $\currentEpoch$, as well as the subsets $\quorum\up{\bp_1+1}, \dots, \quorum\up{\bp_{\ell-1}+1}$ and aggregate the signatures of all validators in these sets (\step{2}, \step{3}).
The full node then sends over (\step{4}) the break points, the combined public keys of all longest-running quora and the most recent block header.
To conclude, the light client verifies that each longest-running quorum signed the next one's keys (\step{5}) and uses the final quorum's keys to verify the most recent header, from which the newest state is obtained (\step{6}, \step{7}).

\subheading{Efficiency.}
Clearly, the communication complexity and verifier work in our protocol are linear in the number $\ell$ of static-quorum periods and independent of $\committeeSize$ due to restricting the light client's checks to the aggregated public key of a valid quorum.
The full node only performs $\missingBlocks |\quorum|$ elliptic curve multiplications to aggregate signatures.

As mentioned in~\cref{obs:trust}, we note that existing light client protocols either put a high computational load on full nodes or require a relaxed trust model to be secure.
Ours, on the other hand, is secure in a strong trust model (analyzed next) and simultaneously obviates the need for full nodes to compute expensive full-fledged zero-knowledge proofs in favor of mere multiplications.

\subsection{Security Analysis}

{{As described in \cref{sec:attackers}, we focus on the case that a subset of validators is malicious and aims to provide the light client with wrong information. Our scheme provides security in that case by efficiently ensuring that the response of the full node is \emph{signed} by the \emph{majority} of all validators. }}
To analyze our construct's security, we first show that a certain building block is implicit in our protocol description, which we refer to as a transitive multi-signature scheme, and for which we can independently prove a meaningful security notion.
Then, the overall security will follow from the transitive multi-signature's security as well as the system's inherent trust assumptions (that is, without introducing new ones).

While transitive signatures are known---recall from \cref{sec:crypto} that they allow for $O(1)$ verification of sequences of messages that are consecutive in a well-defined sense, only based on the start and end points, we are not aware of transitive signatures with multiple signers' public keys having been considered before.
We formally define this generalized functionality in \cref{sec:transitiveMulti} as a transitive multi-signature.

\begin{proposition}\label{prop:multi}
The signature scheme implicitly defined by \cref{eq:signing}, Step~\step{3} and Step~\step{5} is a secure transitive multi-signature scheme.
\end{proposition}

\begin{proof}\emph{(Sketch.)}
{{The claim follows from two facts. (1) The signing, aggregation, and verification functions used in the protocol are the same as those of the scheme in \cref{fig:transitive} (simply substituting the variable $x$ with $i-1$, $y$ with $i$ or $(i, \Nextkeys\up{i})$).
(2) The scheme in \cref{fig:transitive} is a secure transitive multi-signature scheme (which is shown in \cref{prop:transitiveMulti} in~\cref{sec:transitiveMulti}).}}
\end{proof}

By leveraging~\cref{prop:multi} and the per-epoch majority honesty of validators---an assumption already present in most blockchain systems---we can now prove the security of our light client.

\newcommand{\propLightclient}{
The protocol of \cref{fig:pseudocode} is a secure light client update protocol as per \cref{def:UpdSecurity}.
}
\begin{proposition}\label{prop:lcSecurity}
\propLightclient{}
\end{proposition}

\begin{proof}\emph{(Sketch.)}
To satisfy \cref{def:UpdSecurity}, we argue that if the update protocol succeeds, then the light client is at the correct blockchain state in the epoch specified in block header \newestHeader{} (in epoch \currentEpoch{}).
Specifically, if the check in step \step{5} succeeds, then the light client can be sure that the validators in $\quorum\up{\startEpoch-1}$ have indeed signed $\Nextkeys\up{\bp_1}$ etc.
This is due to this check exactly implementing the transitive signature verification, which fulfills exactly this security notion (\cref{prop:multi,prop:transitiveMulti}).

{{For simplicity, let us first consider the case that only one epoch needs to be covered, i.e., $\missingBlocks=1$. The multi-signature property allows to combine signatures from $\quorum\up{\startEpoch-1}$ into one signature. Here, the light client ensures that the validators represented in $\quorum\up{\startEpoch-1}$ correspond to the majority of currently active validators. 
Obviously, an attack can only succeed if an attacker attempts to ''sneak'' wrong values into this process. For example, a malicious full node could send a malformed blockchain header $\Hdr$ ostensibly from epoch $\currentEpoch$.
As honest $\currentEpoch$ validators would not sign a malformed header, and the light client uses the honest validators' public keys to verify the signature over $\Hdr$, our protocol design inherently prevents such attacks.

  Instead of forging signatures, the malicious validators in $\quorum\up{\startEpoch-1}$  could alternatively decide to jointly sign the same wrong input. That is, we would have wrong values that are authenticated by correct signatures. As the remaining validators in $\quorum\up{\startEpoch-1}$ are honest, this would result in a situation where an incorrect message is signed by the malicious validators and a correct (potentially conflicting) message is signed by the honest validators. These (incorrect and correct) messages cannot be aggregated into a single correct signature without causing a mismatch. This mismatch would be easily detected by the light client, which would reject the received state information.
As long as the quorum $\quorum$ of validators accepted by the light client contains (sufficiently many) honest validators, any such attack would fail. 
  
  This leaves attackers only one option: to trick the light client into accepting a wrong quora of validators. That is, the full node and malicious validators could try to include an incorrect $\Nextkeys$ value in $\pi$.
However, because the light client checks that each $\Nextkeys$ is signed by the previous $\Nextkeys$ and the first used $\Nextkeys$ is trusted, there is no opportunity to include a bad $\Nextkeys$ for a malicious full node.}}
We include the full proof in~\cref{sec:proof}.
\end{proof} \section{Implementation \& Evaluation}

In this section, we detail the prototype implementation of our solution, and we evaluate its performance compared to PoPoS~\cite{popos} and CSSV~\cite{cssv}.

\begin{figure*}[tbp]
    \centering
    \begin{subfigure}{0.32\textwidth}
        \includegraphics[width=\linewidth]{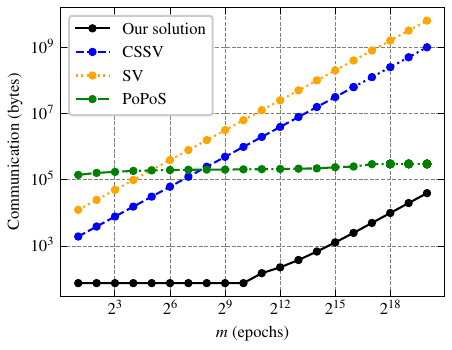}
         \caption{Proof size. For PoPoS: size of all full nodes' messages to the client.\label{fig:evaluationSize}}
    \end{subfigure}
\begin{subfigure}{0.32\textwidth}
        \includegraphics[width=\linewidth]{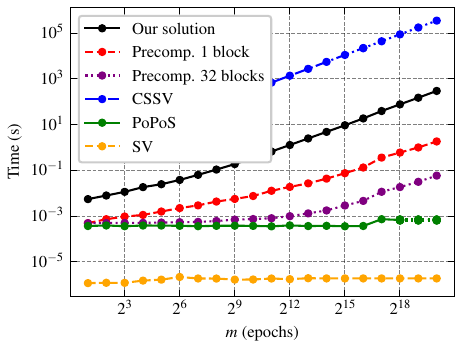}
         \caption{Proof creation time (measured at the full node).\label{fig:evaluationProver}}
    \end{subfigure}
\begin{subfigure}{0.32\textwidth}
        \includegraphics[width=\linewidth]{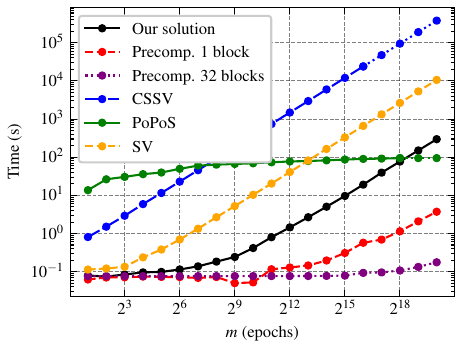}
         \caption{End-to-end update latency (measured client-side).\label{fig:evaluationLatency}}
    \end{subfigure}
    \caption{{Evaluation results of our solution when compared to SV, PoPoS~\cite{popos} and CSSV~\cite{cssv} in our implementation setup.
Also note that the PoPoS code crashed for values of $m>2^{18}$.
    We, therefore, had to extrapolate the PoPoS latencies corresponding to $m=2^{19}$ and $m=2^{20}$.}
\label{fig:evaluation}}
        \vspace{-1 em}
\end{figure*} 
\subsection{Implementation Setup}

We implemented a prototype light client system in Go.
In our implementation, we relied on gnark-crypto's \cite{gnark-crypto} implementation of the elliptic curve BLS12-377\footnote{The curve BLS12-377 is chosen to enable a fair comparison for comparison with CSSV, whose prototype implementation relies on this curve.
Our scheme and PoPoS are agnostic to the curve choice.}.
On this curve, a single signature takes up 96 bytes, and a public key takes up 48 bytes.
Since our main goal was to evaluate the performance of our approach compared to competing approaches, in our prototype implementation, the full node has access to a generated chain of validator public keys and signatures that it feeds to the light client.
Specifically, we generated $2^{20}$ sets of $\committeeSize = 128$ public keys, with each set standing for one epoch.
To model the setting that validator changes are always possible but often have a static quorum, the epochs are divided into subperiods of a fixed number of epochs where such a static quorum is present for the whole subperiod.
The full node also had signatures, valid under these $\quorumSize$ static public keys, over each set of the next $\committeeSize$ keys.
Additionally, we performed measurements where the full node's work is helped by realistic pre-computations: namely, { (1) \enquote{\textbf{precomp., 1 block}}, where the signatures from each epoch have already been pre-aggregated (i.e., for each epoch $i$, $\prod_{v} \sig\down{v}\up{i}$ was already available in step \step{3} of \cref{fig:pseudocode}) and (2) \enquote{\textbf{precomp., 32 blocks}}, where all signatures for 32 consecutive epochs have already been pre-aggregated.
The latter strategy is geared towards settings with short epoch times like in Cosmos (6 seconds), where light client updates over larger $\missingBlocks$ (like $2^{17}$) are expected.}
We deployed our full node implementation on a dual 64-core AMD EPYC 7742 machine, clocked at 2.25 GHz, with 1024 GB of RAM.
The light client ran on an 8-core Apple M2 laptop. 
We conducted our experiments within a wide area network (WAN) with a mean delay of 30 ms and a mean bandwidth of 16 megabits per second. We argue that such a WAN setting faithfully emulates the deployment setting of many light client users around the globe. In our setup, each client invokes an operation in a closed loop; a client may have at most one pending operation.
For different values of $\missingBlocks \leq 2^{16}$, we measure the end-to-end latency observed by the client for a successful update from height 1 to height $\missingBlocks$, as well as the proof generation time measured at the full node.
Each data point in our plots is averaged over 50 independent trials.

\subheading{Baselines.} 
For a baseline comparison, we considered the light client protocols of~\cite{popos} (\enquote{\textbf{PoPoS}}) and of~\cite{cssv} (\enquote{\textbf{CSSV}}), {as well as sequential verification (\enquote{\textbf{SV}}) which sends all public keys of all intermediate validators in full---similarly to Bitcoin's SPV and as used in Polkadot and by Ethereum clients like Helios (Section~\ref{sec:related})}.
We chose PoPoS due to it being the only asymptotically sublinear protocol and CSSV because it seems to be the most efficient one in the same trust model as ours.
For PoPoS, we used the benchmarking code provided alongside the paper \cite{popos} with the following modifications:
\begin{enumerate}
\item While PoPoS is a multi-prover protocol, we run 8 different full nodes on independent CPU cores on the \emph{same} server machine so that the light client has a uniform network connection to all full nodes;
\item We set the committee size parameter $\committeeSize = 128$ as above and set the internal \enquote{tree degree} parameter to 100 which is recommended in \cite{popos};
\item We add code to measure the total time spent by the full node on computing its responses; below, we report the average of the 8 full nodes.
\end{enumerate}

Regarding CSSV's (single-server, like ours) protocol, in lieu of an implemented networked light client benchmark, we used the local evaluation tool part of the open-source implementation of~\cite{cssv} (unmodified).
The tool runs on a single machine to create a chain of validator sets and corresponding signatures and proofs.
The time to create and verify light client proofs is output for each individual validator set.
Below, as \enquote{end-to-end latency} for an update from epoch 1 until epoch $i$, we use the sum of all the proving times plus the verification times of epochs from 1 to $i$.

{{\subheading{Parameter Choice.}}}
Note that, in our experiments, the meaning of absolute $\missingBlocks$ values depends on the actual instantiation: in terms of wall-clock time, its meaning depends on the definition of an epoch in a given system (measured in a number of blocks) as well as the time for blocks to be created.
For example, $2^{16}$ Cosmos epochs mean, with around 6 seconds per block and 1 block per epoch, around 4 days and 12 hours.
On the other hand, in Polkadot epochs, $2^{16}$ would instead amount to around 179 years because an epoch spans one day.

{Notice that the performance of our protocol largely depends on the length of the period during which a static quorum is present. Since our scheme is designed to optimize for rare committee changes, the larger this period, the more efficient our scheme is. In our experiments below, we set this length to 2000\footnote{As shown in \cref{fig:committees}, Cosmos often has longer periods with a static quorum.}.}

\subsection{Performance Evaluation}\label{sec:evaluation}

\noindent We now evaluate the proof size, the latency, and the full node costs incurred by our solution when compared to PoPoS~\cite{popos} and CSSV~\cite{cssv}. 
The full results are depicted in \cref{fig:evaluation}.

{{\subheading{Security.}
First, let us recall that SV, our scheme, and CSSV operate in the standard security model where full nodes are untrusted and only \emph{validators} are the trust anchors.
PoPoS, on the other hand, provides weaker security, where the light client must always be connected to an honest full node. In other words, corrupting a single full node in PoPoS can break the security of the protocol, while the other schemes provide cryptographic security and still can resist a compromise of an arbitrary number of full nodes.}}

\subheading{Proof Size.}
The comparison of proof sizes is shown in \cref{fig:evaluationSize}.
In the parameter regime of this evaluation, our scheme sends a proof with a size of 152 bytes per static-quorum period.
This proof is composed of (1) the new epoch number (8 bytes), (2) an aggregated public key (48 bytes), and (3) the aggregated signature (96 bytes). Since short updates ($\missingBlocks < 2000$) in our implementation setup fall into such periods, this is exactly the number of communicated bytes in these cases.
With larger $\missingBlocks$, proof size increases linearly up to around 80 kilobytes for $2^{20}$ epochs.
On the other hand, PoPoS' interactive multi-round multi-server protocol required more data to be transmitted overall: almost uniformly the same amount of data for all $\missingBlocks$, namely around 200 kilobytes.
Thus, our solution has a smaller proof size by more than 3 orders of magnitude when considering a single static-validator period---a gap which narrows as $\missingBlocks$ increases.
{For SV and CSSV, since static-committee periods are not optimized, every epoch has a (redundant) proof of size 6240 and around 976 bytes, respectively.
In direct comparison, our solution has a smaller size, compared to CSSV, by at least $5.8\times$ and in the experiments with $\missingBlocks=2^{20}$ by 4 orders of magnitude. When compared to the basic SV baseline, our solution exhibits smaller proof sizes by 5 orders of magnitude.}

\subheading{Full Node Computations.}
The comparison of proof creation costs is shown in \cref{fig:evaluationProver}.
With our solution, the full node's computations consist of $\missingBlocks \cdot \committeeSize = \missingBlocks \cdot 128$ elliptic curve multiplications. In our setup, this took under 100 milliseconds for $\missingBlocks<2^8$ and around 700 seconds for $\missingBlocks=2^{20}$.
The 1-block pre-computation strategy needs only $\missingBlocks$ multiplications since the first aggregation step over $\committeeSize$ has been done previously.
\cref{fig:evaluationProver} confirms that the 1-block pre-computations reduce the full node's computation time to under 100 milliseconds when $\missingBlocks \leq 2^{15}$. 
{The 32-block pre-computation strategy further reduces this overhead by a factor of 32.}
{In the case of SV, the overhead of proof generation is negligible since the role of the full node simply consists of relaying the signatures gathered from the validators.
Moreover, PoPoS full nodes have lower requirements since the protocol's core operations are logarithmic in $\missingBlocks$ (albeit at the cost of additional trust assumptions; see Section~\ref{sec:related}).}
We contrast this with CSSV, where computing a proof took around 200 milliseconds \emph{for one epoch}.
For instance, we note that, in the time taken to compute a CSSV proof for one epoch, our solution with pre-computations allows creating a proof that covers over $2^{16}$ epochs, a $65536\times$ improvement or 4 orders of magnitude.
When aggregating from scratch, the improvement is still $2048\times$.

\subheading{Latency.}
The comparison of end-to-end client latencies is shown in \cref{fig:evaluationLatency}.
Note that these values are lower bounded by prover computation time plus network latency (except for CSSV, where the numbers have no networking component).
We can thus conclude from the prover time shown in \cref{fig:evaluationProver} alone that our solution's latency is likely dominated by the implementation setup's round trip time of 30 milliseconds for $\missingBlocks$ as large as 32 epochs or, with pre-computations, for $\missingBlocks$ as large as $2^{13}$.
Indeed, as seen in \cref{fig:evaluationLatency}, for smaller $\missingBlocks$, our solution's end-to-end latency is close to constant.
After that, we find linear increases in accordance with the increasing prover time.
Overall, our scheme's latency is below one second for $\missingBlocks$ up to $2^{11}$ or, with pre-computations, up to $2^{20}$.
We contrast this with PoPoS where, likely due to the multiple rounds of network communications, even an update for $\missingBlocks=2$ epochs takes over 20 seconds.

{
Asymptotically, PoPoS' latency scales attractively as the case $\missingBlocks=2^{20}$ still takes only around 100 seconds.
When compared to the non-pre-computation variant of our solution (300 seconds), this is favorable.
However, two caveats apply.
First, as we observed, large $m$ updates only make up a relatively small fraction of real-world requests (see \cref{obs:m}).
To put the absolute values into perspective, in a Polkadot-type system with 1 epoch per day, $m \geq 2^{18}$---where PoPoS outperforms our scheme's non-pre-computation variant---is equivalent to over 700 years, which is unrealistic.
Such large update distances can only occur in systems with short epochs akin to Cosmos.
But here, as the second caveat, the picture changes when pre-computations are considered.
Indeed, the 1-block pre-computation strategy already achieves a performance improvement of more than an order of magnitude compared to PoPoS at $m=2^{20}$.
The 32-block pre-computation strategy further reduces the overall latency by approximately 3 orders of magnitude compared to PoPoS at $m=2^{20}$.

Finally, we observe that our solution's end-to-end latency, both with and without pre-computations, already at $\missingBlocks=2$, is around $10\times$ lower than CSSV's and half of SV's.
These gaps increase, leading to improvements of, respectively, 3 orders of magnitude and $30\times$ without any pre-computations.
With pre-computations, the difference to CSSV and to SV is around 6 and, respectively, 4 orders of magnitude.}

 \section{Conclusion}

In this paper, we presented a novel light client protocol for committee-based blockchains that exhibits low communication and computation costs on both the lightweight client and full nodes. Our protocol is mostly geared to practical deployment settings where (1) lightweight clients are not offline for long periods of time or (2) the signers of block headers have a long-term static core quorum.
Our protocol leverages transitive signatures \cite{transitiveSignatures} to verify long sequences of consecutive block headers---independently of the number of blockchain epochs.
To improve efficiency further, we argue that the compatibility of transitive signatures with existing multi-signature aggregation techniques allows for better compression of signatures issued by the blockchain validators. 
We validated the performance of our protocol by means of prototype implementation in a realistic WAN setting. Our evaluation results show that our protocol achieves a reduction by up to $41726\times$ in end-to-end update latency and around $1000\times$ smaller proof size when compared to the PoPoS~\cite{popos} and CSSV~\cite{cssv} designs.
 
\section*{Acknowledgements}
\noindent
The authors would like to thank the anonymous shepherd and reviewers for their constructive feedback.
We also thank Anna Piscitelli and Annika Wilde for testing and discussing our evaluation code and its documentation.
We thank Kyle Rudnick and Julian Willingmann for assistance in gathering data for \cref{fig:m,fig:committees}.

This work has been partially funded by the Deutsche Forschungsgemeinschaft (DFG, German Research Foundation) under Germany's Excellence Strategy - EXC 2092 CASA - 390781972 and by the European Union through the HORIZON-JU-SNS-2022 NANCY project with Grant Agreement number 101096456. Views and opinions expressed are, however, those of the author(s) only and do not necessarily reflect those of the European Union or the SNS JU. Neither the European Union nor the granting authority can be held responsible for them.

\printbibliography[heading=bibintoc]
\normalsize
\appendices
\section{Transitive multi-signatures}\label{sec:transitiveMulti}

In this section, we put forth a notion of transitive multi-signatures.
We show that constructing a transitive multisiganture from Bellare and Neven's GapTS-2 \cite[Section 5.C]{bellareNevenTransitive} scheme works analogously to the construction of multi-signatures from BLS signatures.

\subsection{Definitions}\label{sec:proofDefs}

A standard transitive signature consists of three algorithms \KeyGen{}, \Sign{} and \Verify{} and makes it possible to combine the same signer's signatures over messages of the form $(x,\allowbreak{} y)$ and $(y,\allowbreak{} z)$ into a valid signature over $(x,\allowbreak{} z)$.
Such messages are best interpreted as edges of an undirected graph $G$, meaning that aggregation is possible along \emph{paths}.
The formal security notion for transitive signatures is dubbed transitive unforgeability under adaptive chosen message attack (tuf-cma) and demands that no probabilistic polynomial-time adversary with access to a signing oracle for the public key $\pk$ can, with non-negligible probability, output $x,\allowbreak{} y$ and a signature $\sig$ such that (1) $\Verify(\pk,\allowbreak{} x,\allowbreak{} y,\allowbreak{} \sig)$ accepts even though (2) there is no path from $x$ to $y$ in $G$.
Here, queries to the signing oracle add nodes and edges to $G$.

\subheading{Transitive Multi-Signatures.}
The following definition is inspired by existing notions of signature aggregation \cite[cf.][]{thresholdBLS}, which allows signatures from different sets of signers on different messages to be combined into a single one. 
In practical systems, it is typical to demand a key registration phase in which so-called proofs of possession are checked in order to prevent certain attacks \cite[cf.][]{ry07}.
For simplicity, our analysis will model this usage using a so-called knowledge of secret key assumption (see below), and we hence omit explicitly specifying proof of possession algorithms as part of the signature scheme definition\footnote{For the scheme presented below, the same techniques as for BLS signatures apply.}.

\begin{definition}[Transitive Multi-Signature]
A transitive multi-signature consists of the following algorithms.
\begin{itemize}
    \item $\KeyGen(\secparam) \rightarrow (\sk, \pk)$.
        For a security level $\secpar$, outputs a secret key and a public key.
    \item $\Sign(\sk, x, y) \rightarrow \sig$.
        For $x, y \in \{0,1\}^*$, outputs a signature $\sig$.
    \item $\Aggregate(\sig_1, \dots, \sig_n) \rightarrow \Sig$.
        For signatures $\sig_1,\allowbreak{} \dots,\allowbreak{} \sig_n$, outputs an aggregated signature $\Sig$.
    \item $\Verify((\PK_1, \dots, \PK_n),\allowbreak{} ((x_1,\allowbreak{} z_1),\allowbreak{} \dots,\allowbreak{} (x_n,\allowbreak{} z_n)),\allowbreak{} \Sig) \rightarrow 0/1$.
        For lists of public keys $\PK_1,\allowbreak{} \dots,\allowbreak{} \PK_n$, for $x_1,\allowbreak{} z_1,\allowbreak{} \dots,\allowbreak{} x_n,\allowbreak{} z_n \in \{0,1\}^*$ and for a signature $\Sig$, accepts or rejects.
\end{itemize}

A transitive multi-signature supports \emph{secure aggregation} if no probabilistic polynomial-time adversary with access to a signing oracle for the public key $\pk$ can, with non-negligible probability, output sets of public keys $\PK_1,\allowbreak{} \dots,\allowbreak{} PK_n$, a signature $\sig$ and $(x_1,\allowbreak{} z_1),\allowbreak{} \dots,\allowbreak{} (x_n,\allowbreak{} z_n)$ such that
\begin{enumerate}
\item $\Verify((\PK_1,\allowbreak{} \dots,\allowbreak{} \PK_n),\allowbreak{} ((x_1,\allowbreak{} z_1),\allowbreak{} \dots,\allowbreak{} (x_n,\allowbreak{} z_n)),\allowbreak{} \sig)$ accepts,
\item $\exists j \in [n]$ such that $\pk \in PK_j$ and
\item there is no path between $x_j$ and $z_j$ in the graph created from signing queries to the $\pk$ oracle.
\end{enumerate}
\end{definition}

\subsection{Construction}

\newcommand{\belowLine}{4pt}
\newcommand{\nextKeysHeight}{\vphantom{$\Nextkeys\up{i-1}$}}
\begin{figure}[!t]
\centering
\begin{pcvstack}[boxed, space=\baselineskip]
    \begin{pchstack}[space=0.3\baselineskip]
       \procedure[bodylinesep=\belowLine]{$\KeyGen(1^\lambda)$}{
            \sk \sample \ZZ_p \\
            \pk = g^{\sk} \\
            \pcreturn (\sk, \pk)
       }
        \procedure[bodylinesep=\belowLine]{$\Sign(\sk, x, y)$}{\sig = \left(\frac{\groupHash(y)}{\groupHash(x)}\right)^\sk \\
            \pcreturn \sig
        }
        \procedure[bodylinesep=\belowLine]{$\Aggregate(\sig_{1}, \dots, \sig_{n})$}{\Sig = \prod_{i \in [n]} \sig_{i} \\
            \pcreturn \Sig
        }
    \end{pchstack}
    \begin{pchstack}[space=0.5\baselineskip]
\procedure[bodylinesep=\belowLine]{$\Verify((\PK_{1}, \dots, \PK_{n}), ((x_1,z_1),\dots,(x_n,z_n)), \Sig)$}{\text{for}\ i \in [n], \apk_{i} = \prod_{\pk \in \PK_{i}} \pk \\
            \text{Check}\ e(g, \Sig) = \prod_{i \in [n]} e\left(\apk_{i}, \frac{\groupHash(z_i)}{\groupHash(x_i)}\right)
        }
    \end{pchstack}
\end{pcvstack}
    \caption{
        Transitive multi-signature scheme.
\label{fig:transitive}
    }
\end{figure} The scheme based on GapTS-2 is shown in \cref{fig:transitive}.
GapTS-2 itself consists of similar key generation and signing algorithms as in \cref{fig:transitive} and only restricts $n=1$ and $|\PK_1|=1$ in \Verify{}.
It is shown in \cite{bellareNevenTransitive} that GapTS-2 is a secure standard transitive signature scheme under the one-more computational Diffie-Hellman assumption in the random oracle model.
Based on this, we now prove that the multi-signature is secure as well.

\begin{proposition}\label{prop:transitiveMulti}
    The scheme of \cref{fig:transitive} is a secure transitive multi-signature.
\end{proposition}

\begin{proof}
    We follow similar steps as in the security proof for BLS multi-signatures from \cite[Section 4.3]{thresholdBLS}.
    That is, we assume the existence of an adversary $\advA$ against aggregation security and will derive a tuf-cma adversary $\advB$ (\cref{sec:proofDefs}) against GapTS-2.

    For simplicity and ease of presentation, we make a \emph{knowledge of secret key} assumption whereby the attacker must directly output the secret keys belonging to the (other) public keys in its output besides $\pk$.
    This models a (real-world) key registration model, where verification would additionally require each public key to have a separate \emph{proof of knowledge} of the corresponding secret key \cite[cf.][]{ry07}.
    Then, since an attacker $\advA$ would also have to prove knowledge of all the secret keys it used, by a standard definition of a proof of knowledge \cite{lindell03}, there must exist a polynomial-time extractor algorithm that would interact with $\advA$ to obtain the secrets.
    The explicit knowledge assumption abstracts this step. 

    Now, suppose $\advB$ received the public key $\pk$ in the tuf-cma experiment.
    $\advB$ would send all parameters to $\advA$, forward all of $\advA$'s queries to the tuf-cma challenger, and forward the responses to $\advA$.
    Eventually, $\advA$ would output $\PK_1, \dots, \PK_n, x_1, z_1 \dots, x_n, z_n$ and $\sig$.
    Say $j$ is such that $\pk \in \PK_j$.
    Let $\SK_1, \dots, \SK_n$ be the lists of secret keys corresponding to the public keys, which $\advB$ could obtain as per the knowledge of secret key assumption---except for the secret key $\sk^*$ corresponding to $\pk$.
    If $\sig$ passed verification, then it would hold that
    \[
    \sig = \prod_{i \in [n]} \left(\frac{\groupHash(z_i)}{\groupHash(x_i)}\right)^{\sum_{\sk \in \SK_i}}.
    \]
    Hence, $\advB$ divides $\sig$ by
    \[
    \left(\prod_{i \in [n]\setminus\{j\}} \left(\frac{\groupHash(z_i)}{\groupHash(x_i)}\right)^{\sum_{\sk \in \SK_i} \sk}\right) \cdot \left(\frac{\groupHash(z_j)}{\groupHash(x_j)}\right)^{\sum_{\sk \in \SK_j\setminus\{\sk^*\}} \sk}
    \]
    and obtains $\left(\groupHash(z_j)/\groupHash(x_j)\right)^{\sk^*}$.
    This constitutes a valid forgery for GapTS-2, which contradicts tuf-cma security and completes our proof.
\end{proof} \section{Light Client Security}\label{sec:proof}

\subheading{\cref{prop:lcSecurity}.}
\propLightclient{}

\begin{proof}

To formally map the security definition of \cref{sec:model} to the protocol description, let us first make explicit that the light client's local storage $\lcStateOld$ at a height $\startHeight$ consists of the blockchain state digest (that is, the digest of $\BCstateOld$), the blockchain epoch number $\startEpoch$ and the aggregated public key $\Nextkeys\up{\startEpoch-1}$ of the validator quorum $\quorum\up{\startEpoch}$.

Next, we convince ourselves that the validators, as per \cref{eq:signing}, are implicitly using exactly the transitive signature scheme for which we proved secure aggregation in \cref{sec:transitiveMulti}.
Specifically, the inputs to \Sign/ in epoch $i$ are $(i-1, i)$ or, if $i$ is a break point, $(i-1, (i, \Nextkeys\up{i}))$.
Likewise, the full nodes operations in Step \step{3} are equivalent to the \Aggregate{} algorithm and the light client performs, in Step \step{5}, a version of the \Verify{} algorithm where the quantity denoted as $\apk_1$ in \cref{sec:transitiveMulti} is already given as a precomputed value, namely $\Nextkeys\up{\startEpoch-1}$.
By the validators' honest majority assumption (since they are the ones who originally supply this precomputed value), it is secure to use this value.

Thus, we can rely on the security properties of the signature shown in \cref{prop:transitiveMulti}.
In particular, we are now ready to prove what our security definition requires:
that when verification passes, then the final state is a correct state except with negligible probability.

\begin{itemize}
    \item Based on the signature property proven in \cref{sec:transitiveMulti}, if the check in step \step{5} passes, then the light client can be sure that sufficiently many validators in $\validators\up{\startEpoch}$ have individually signed the chain of messages $\startEpoch, \startEpoch+1, \dots, \bp_1-1, (\bp_1, \Nextkeys\up{\bp_1})$---except with negligible probability.
    \item The same holds for the subsequent subperiods.
    \item Due to majority-honesty, if enough validators signed, then the aggregated public keys at the break points are all the correct ones.
    \item If, in particular, the final provided set of keys---that is, $\Nextkeys\up{\bp_\ell}$ or $\Nextkeys\up{\bp_{\ell-1}}$---is correct, then it can safely be used to verify the header $\Hdr$.
    \item Finally, the state read off of a correct block header must be a correct state.
\end{itemize}

Thus, for our check to pass even though the final state was a false one, one of our assumptions (round-wise majority-honesty of validators or siganture security) would have to be broken.
This completes the proof.
\end{proof} {\section{Data from \cref{fig:m,fig:committees}}

\newcommand{\dataWidth}{0.58cm}

\begin{table}[tbp!]
    \centering
    \caption{
        Distribution of update distances requested by light clients (\cref{fig:m}).
        For example, we find that 30\% light clients were already at the same height as our full node (first column).
        \label{tab:cum}
    }
    \begin{tabularx}{\linewidth}{l|p{\dataWidth}p{\dataWidth}p{\dataWidth}p{\dataWidth}p{\dataWidth}p{\dataWidth}p{\dataWidth}}
        \toprule
        Quantile & 30\% & 35\% & 40\% & 45\% & 50\% & 55\% & 60\% \\
        $m$ & 0 & 1 & 5 & 11 & 38 & 133 & 307 \\
        \midrule
        Quantile & 65\% & 70\% & 75\% & 80\% & 85\% & 90\% & 95\% \\
        $m$ & 1440 & 12477 & 53418 & 322349 & 834781 & 835259 & 835658 \\
        \bottomrule
    \end{tabularx}
\end{table}

For completeness, we show in \cref{tab:cum} the distribution of update distances ($m$) from all 19458 light client requests observed by our Bitcoin full node over the 2-week measurement period.
This corresponds to the same data which is plotted in \cref{fig:m} and discussed in \cref{sec:observations}.

\newcommand{\committeeWidth}{0.16cm}
\begin{table}[H]
    \centering
    \caption{
        Largest validator subsets which were active for consecutive epochs during our measurement period (\cref{fig:committees}).
        For example, after the first Polkadot committee change we observed, a set of 286 validators from the first epoch remained in the validator set in the next epoch.
        For each chain, the top and bottom rows are subdivisions of the 2-week period into two of length one week.
        For XRP Ledger and Cosmos, we only show a subsample of the 20000 and, respectively, 40 data points due to space constraints.
        \label{tab:quorumsize}
    }
    \begin{tabularx}{\linewidth}{X|p{\committeeWidth}p{\committeeWidth}p{\committeeWidth}p{\committeeWidth}p{\committeeWidth}p{\committeeWidth}p{\committeeWidth}p{\committeeWidth}p{\committeeWidth}p{\committeeWidth}p{\committeeWidth}}
        \toprule
        \multirow{2}{*}{Cosmos} & 180 & 180 & 180 & 180 & 168 & 168 & 168 & 168 & 167 & 167 & 167 \\
                                & 180 & 180 & 179 & 179 & 178 & 177 & 175 & 175 & 175 & 175 & 175 \\
        \midrule
        \multirow{2}{*}{Polkadot} & 286 & 281 & 276 & 273 & 272 & 272 & 266 & 265 & 247 & 247 & 245 \\
                                  & 278 & 273 & 270 & 268 & 267 & 267 & 266 & 265 & 263 & 263 & 288 \\
        \midrule
        \multirow{2}{*}{XRP Ledger} & 35 & 35 & 35 & 35 & 35 & 35 & 35 & 35 & 35 & 35 & 35 \\
                                    & 35 & 35 & 35 & 35 & 35 & 35 & 35 & 35 & 35 & 35 & 35 \\
        \bottomrule
    \end{tabularx}
\end{table}

In \cref{tab:quorumsize}, we additionally show the data regarding static validator subsets in Cosmos, Polkadot and XRP Ledger we collected from public nodes over a 2-week period (plotted in \cref{fig:committees} and discussed in \cref{sec:observations}).}

\end{document}